\def\notes{0}

\documentclass[11pt]{article}
\usepackage[margin=1in]{geometry}

\usepackage{amsmath}
\usepackage{amssymb}
\usepackage{amsthm}
\usepackage{fullpage}
\usepackage{mathpazo}
\usepackage{mathabx}
\usepackage[pagebackref,colorlinks=true,pdfpagemode=none,linkcolor=blue,citecolor=blue,pdfstartview=FitH]{hyperref}
\usepackage[nameinlink,capitalize]{cleveref}
\crefname{claim}{Claim}{Claims}

\newtheorem{theorem}{Theorem}[section]

\newtheorem{lemma}[theorem]{Lemma}

\newtheorem{proposition}[theorem]{Proposition}
\newtheorem{definition}[theorem]{Definition}
\newtheorem{claim}[theorem]{Claim}

\newtheorem{remark}[theorem]{Remark}

\newcommand{\X}{\mathcal{X}}
\newcommand{\Y}{\mathcal{Y}}
\newcommand{\Z}{\mathcal{Z}}
\newcommand{\ve}{{\varepsilon}}
\newcommand{\reals}{\mathbb{R}}
\renewcommand{\epsilon}{\varepsilon}
\renewcommand{\phi}{\varphi}
\newcommand{\E}{\mathbb{E}}
\newcommand{\I}{\mathbb{1}}
\newcommand{\B}{\mathcal{B}}
\newcommand{\cT}{{\cal T}}
\newcommand{\eps}{\varepsilon}

\newcommand{\etal}{\emph{et al.\/}}
\newcommand {\alice} {{\sf {Alice}}}
\newcommand {\bob} {{\sf {Bob}}}
\newcommand{\prt}{{\mathsf{prt}}}
\newcommand{\qprt}{{\mathsf{qprt}}}
\newcommand{\rprt}{{\mathsf{rprt}}}

\newcommand {\srec} {\mathsf{srec}}
\newcommand{\CC}{\mathsf{CC}}
\newcommand{\IC}{\mathsf{IC}}
\newcommand{\QC}{\mathsf{QC}}
\newcommand{\adv}{\mathsf{adv}}
\newcommand{\ceil}[1]{\left\lceil #1 \right\rceil}
\newcommand{\OPT}{\mathsf{OPT}}
\newcommand{\subcube}{\mathsf{subcube}}
\newcommand{\size}{\mathsf{size}}

\newcommand{\support}{\mathsf{support}}
\renewcommand{\Bar}{\widebar}
\newcommand{\poly}{\text{poly}}

\ifnum\notes=0
\usepackage[disable]{todonotes}
\else
\usepackage[colorinlistoftodos]{todonotes}
\fi

\newcommand{\changes}{}

\newcommand{\lref}[2][]{\hyperref[#2]{#1~\ref*{#2}}}
\renewcommand{\eqref}[1]{\hyperref[#1]{(\ref*{#1})}}
\numberwithin{equation}{section}

\newcommand {\Rahul} {Rahul Jain}
\newcommand {\Prahladh} {Prahladh Harsha}
\newcommand {\Jaikumar}{Jaikumar Radhakrishnan}
\newcommand{\CQTCS}{CQT, MajuLab and NUS, Singapore. }
\newcommand {\TIFR}{TIFR, Mumbai.}
\newcommand{\simons}{Work done when the three authors were visiting
  the Simons Institute for the Theory of Computing, Berkeley,
  USA.}

\newcommand {\Title} {Partition bound is quadratically tight for product distributions} 
\date{}

\title{\Title\footnote{\simons}}
\author{
\Prahladh\thanks{\TIFR\ {\tt prahladh@tifr.res.in}. Research
    supported in part by ISF-UGC grant 1399/4 and Google India Fellowship.}
\and 
\Rahul\thanks{\CQTCS\ {\tt rahul@comp.nus.edu.sg}. Partly supported by
  the Singapore Ministry of Education via Academic Research Fund Tier 3 MOE2012-T3-1-009 
and Young Researcher Award, National University of Singapore.}  
\and
\Jaikumar\thanks{\TIFR\ {\tt jaikumar@tifr.res.in}.} \\}

\begin{document}
\begin{titlepage}

\maketitle

\thispagestyle{empty}
\setcounter{page}{0}

\begin{abstract}


  Let $f : \{0,1\}^n \times \{0,1\}^n \rightarrow \{0,1\}$ be a
  2-party function. For every product distribution
  $\mu$ on $\{0,1\}^n \times \{0,1\}^n$, we show that 
$$\CC^\mu_{0.49}(f) =
  O\left(\left(\log \prt_{1/8}(f) \cdot \log \log \prt_{1/8}(f)\right)^2\right),$$
  where $\CC^\mu_\ve(f)$ is the distributional communication
  complexity of $f$ with error at most $\ve$ under the distribution
  $\mu$ and $\prt_{1/8}(f)$ is the {\em partition bound} of $f$, as
  defined by Jain and Klauck [{\em Proc. 25th CCC}, 2010]. We also prove a similar
  bound in terms of $\IC_{1/8}(f)$, the {\em information complexity}
  of $f$, namely,
$$\CC^\mu_{0.49}(f) =
O\left(\left(\IC_{1/8}(f) \cdot \log \IC_{1/8}(f)\right)^2\right).$$
The latter bound was recently and independently established by
Kol [{\em Proc. 48th STOC}, 2016] using a different technique.

We show a similar result for query complexity under product
distributions. Let $g : \{0,1\}^n \rightarrow \{0,1\}$ be a function.
For every bit-wise product distribution $\mu$ on $\{0,1\}^n$, we show
that 
$$\QC^\mu_{0.49}(g) =
  O\left(\left( \log \qprt_{1/8}(g) \cdot \log \log\qprt_{1/8}(g) \right)^2 \right),$$
  where $\QC^\mu_{\ve}(g)$ is the distributional query complexity of
  $f$ with error at most $\ve$ under the distribution $\mu$ and
  $\qprt_{1/8}(g))$ is the {\em query partition bound} of the
  function $g$.

  Partition bounds were introduced (in both
  communication complexity and query complexity models) to provide
  LP-based lower bounds for randomized communication complexity and
  randomized query complexity. Our results demonstrate that these
  lower bounds are polynomially tight for {\em product}
  distributions.

\end{abstract}
\end{titlepage}

\section{Introduction}

Over the last decade, several lower bound techniques using linear
programming formulations and information complexity methods have been
developed for problems in communication complexity and query
complexity. One of the central questions in communication complexity
is to understand the tightness of these lower bound techniques. For
instance, over the last few years, considerable effort has gone into
understanding the {\em information complexity} measure. Informally
speaking, (internal) information complexity is the amount of
information the two parties reveal to each other about their
respective inputs while computing the joint function. It is known that
for product distributions, the internal information complexity not
only lower bounds but also upper bounds the distributional
communication complexity (up to logarithmic multiplicative
factors in the communication complexity)~\cite{BarakBCR2013}. On the other hand, recent works due to Ganor, Kol and
Raz~\cite{GanorKR2014,GanorKR2015,GanorKR2016} show that there exist non-product
distributions which exhibit exponential separation between internal
information complexity and distributional communication
complexity\footnote{The third result of Ganor, Kol and
  Raz~\cite{GanorKR2016} actually demonstrates an exponential
  separation between external information and communication
  complexity, albeit not for computing a Boolean function.}. However, it is still open if internal information
complexity (or a polynomial of it) upper bounds the public-coin
randomized communication complexity (up to logarithmic multiplicative
factors in the input size)~\cite{FontesJKLLR2015}.

Jain and Klauck~\cite{JainK2010}, using tools from linear programming,
gave a uniform treatment of several of the existing lower bound
techniques and proposed the {\em partition bound}. This leads to
following related (but incomparable) conjecture: does a polynomial of
the partition bound yield an upper bound on the communication complexity? We are
not aware of any counterexample to this conjecture\footnote{The recent
work of G\"o\"os~et al.~\cite{GoosJPW2015} demonstrates the existence
of a total function for which the partition bound is strictly sublinear in the
randomized communication complexity. This still does not rule out
communication complexity being bound by a polynomial of the partition bound.}.

We consider these questions when the inputs to Alice and Bob are drawn
from a product distribution and show the following.

\begin{theorem}\label{thm:ccquad} Let $f : \{0,1\}^n \times \{0,1\}^n
  \rightarrow \{0,1\}$, and let $\IC_{\ve}(f)$ and $\prt_{\ve}(f)$ be
  the information complexity and 
  partition bound respectively of $f$ with error at most $\ve$. For a product
  distribution $\mu$ on $\{0,1\}^n \times \{0,1\}^n$, the
  distributional communication complexity of $f$ under distribution
  $\mu$ with error at most $0.49$, denoted by $\CC^\mu_{0.49}(f)$, can
  be bounded above as follows:\changes
\begin{align}
\label{eq:IC}\CC^\mu_{0.49}(f) &=O\left(\left(\IC_{1/8}(f) \cdot \log 
\IC_{1/8}(f)\right)^2\right),\\
\label{eq:prt}\CC^\mu_{0.49}(f) &=O\left(\left(\log \prt_{1/8}(f) \cdot \log \log 
\prt_{1/8}(f)\right)^2\right).
\end{align}
\end{theorem}

Our technique yields bounds more general than those stated above (see
discussion after \cref{prop:sr_rprt} for this generalization). We
remark that recently (and independently of this work)
Kol~\cite{Kol2016} obtained the bound~\eqref{eq:IC} using very
different techniques. Kol's result is stronger in the sense that her
bound is in terms of the information complexity $\IC^\mu(f)$ for the
product distribution $\mu$, while our result is in terms of the worst
case information complexity $\IC(f)$ (note, $\IC_\ve(f)
= \max_{\mu} \IC^\mu_\ve(f)$). In fact, Kol showed that
\begin{equation}
\CC^{\mu}_{\delta+\epsilon}(f) =
O\left(\IC^{\mu}_\delta(f)^2\cdot \poly\log \IC^{\mu}_\delta(f)/\epsilon^5\right), \label{eq:kol-strong}
\end{equation}
and concluded that  
\begin{equation}
\CC^{\mu}_{0.49}(f)
= O\left(\IC_{1/8}(f)^2\cdot \poly\log \IC_{1/8} (f)\right).
\label{eq:kol-weak}
\end{equation}
Kol's result~\eqref{eq:kol-strong} is incomparable to our
second result in terms of partition bound~\eqref{eq:prt}.

We consider a similar question in query complexity and show the following. 
\begin{theorem}\label{thm:qcquad} Let $g : \{0,1\}^n \rightarrow
  \{0,1\}$ be a function and $\mu$ be a bit-wise product distribution
  on $\{0,1\}^n$. Let $\qprt_{\ve}(g)$ be the query partition
  bound for $g$ with error $\ve$. Then, the distributional query
  complexity with error at most $0.49$ under the distribution $\mu$,
  denoted by $\QC^\mu_{0.49}(f)$, can be bounded above as follows:
$$\QC^\mu_{0.49}(g) = O\left(\left(\log \qprt_{1/8}(g) \cdot \log \log\qprt_{1/8}(g) \right)^2\right) . $$  
   \end{theorem}

A similar quadratic upper bound for query complexity for product
distributions in terms of approximate certificate complexity was obtained
by Smyth~\cite{Smyth2002}. His proof uses Reimer's inequality while our
proof technique is based on Nisan and Wigderson's~\cite{NisanW1995}
more elementary approach.

\subsection*{Organization}
The rest of the paper is devoted to the proofs of these two
theorems. The communication complexity result is proven in
\cref{sec:cc} while the query complexity result is proved
in \cref{sec:q}.

\section{Communication Complexity}
\label{sec:cc}

\subsection{Preliminaries}
We work in Yao's two-party communication
model~\cite{Yao1979} (see Kushilevitz and
Nisan~\cite{KushilevitzNisan} for an excellent introduction to the
area). Let $\X$, $\Y$ and $\Z$ be finite non-empty sets, and let $f:
\X \times \Y \to \Z$ be a function. A two-party protocol for computing
$f$ consists of two parties, {\alice} and {\bob}, who get inputs $x
\in \X$ and $y \in \Y$ respectively, and exchange messages in order to
compute $f(x,y) \in \Z$ (using shared randomness).

For a distribution $\mu$ on $\X \times \Y$, let the $\epsilon$-error
distributional communication complexity of $f$ under $\mu$ (denoted by
$\CC_{\epsilon}^{\mu}(f)$), be the number of bits communicated (for the
worst-case input) by the best deterministic protocol for $f$ with
average error at most $\epsilon$ under $\mu$.  Let $\CC^{{\text{\rm pub}}}_{\epsilon}(f)$, the public-coin randomized communication
complexity of $f$ with worst case error $\epsilon$, be the number of
bits communicated (for the worst-case input) by the best public-coin
randomized protocol that for each input $(x,y)$ computes $f(x,y)$
correctly with probability at least $1-\epsilon$. Randomized and
distributional complexity are related by the following special case of
von Neumann's minmax principle.
\begin{theorem}[Yao's minmax principle~\cite{Yao1983}]
\label{thm:yao}
$\CC^{\text{\rm pub}}_{\epsilon}(f) = \max_{\mu} \CC_{\epsilon}^{\mu}(f) $.
\end{theorem}
We will prove \cref{thm:ccquad} by first showing an
upper bound on communication complexity in terms of the smooth
rectangle bound and then observing that the smooth rectangle bound is
bounded above by the partition bound.

\paragraph*{Smooth rectangle bound:} The smooth rectangle bound was
introduced by Jain and Klauck~\cite{JainK2010} as a generalization of
the rectangle bound. Just like the rectangle bound, the smooth
rectangle bound also provides a lower bound for randomized
communication complexity.  Informally, the smooth rectangle bound for
a function $f$ under a distribution $\mu$, is the maximum over all
functions $g$ , which are close to $f$ under the distribution $\mu$,
of the rectangle bound of $g$.  However, it will be more convenient
for us to work with the following linear programming
formulation. (See \cite[Lemma~2]{JainK2010} and
\cite[Lemma~6]{JainY2012} for the relations between the LP formulation
and the more ``natural'' formulation in terms of rectangle bound.) It
is evident from the LP formulation that the smooth rectangle bound is
a further relaxation of the partition bound (defined in the
appendix). We will formulate our results in terms of a distributional
version of the above smooth rectangle bound. For $\mu:\X\times \Y \to
\reals$ and any $z \in \Z$ and rectangle $R$, let $\mu_z(R) : = \mu(R
\cap f^{-1}(z))$ and $\mu_{\bar{z}}(R) := \mu(R) -
\mu_z(R)$. Furthermore, let $\mu_z : = \mu_z(\X \times \Y)$ and
$\mu_{\bar{z}} : = \mu_{\bar{z}}(\X \times \Y)$. The smooth rectangle
and its distributional version are defined below.

\begin{definition}[Smooth rectangle bound]\quad \\

\begin{itemize}
\item For a function $f: \X \times \Y \rightarrow \Z$ and $\ve \in (0,1)$, 
the $(\epsilon,\delta)$-smooth rectangle bound of $f$ denoted
$\srec_{\epsilon,\delta}(f)$ is defined to be
$\max\{\srec^z_{\epsilon,\delta}(f): z\in\Z\}$, where
$\srec^z_{\epsilon,\delta}(f)$ is the optimal value of the
following linear program.
\item For a distribution $\mu$ on $\X\times \Y$  and function $f:
\X \times \Y \to \Z$, the $(\epsilon,\delta)$-smooth
rectangle bound of $f$ with respect to $\mu$ denoted
$\srec^\mu_{\epsilon,\delta}(f)$ is defined to be
$\max\{\srec^{z,\mu}_{\epsilon,\delta}(f): z\in\Z\}$, where
$\srec^{z,\mu}_{\epsilon,\delta}(f)$ is the optimal value of
the following linear program.
\end{itemize}
\quad\\

{\scriptsize
\begin{minipage}{3in}
    \centerline{\underline{$\srec^z_{\eps,\delta}(f)$}}
    \begin{align*}
      \min \quad  \sum_{R }  w_{R} \\
      \sum_{R: (x,y) \in R} w_{R} & \geq 1 - \epsilon, && \forall (x,y) \in  f^{-1}(z)\\
      \sum_{R: (x,y) \in R} w_{R} &\leq \delta,&& \forall (x,y) \not\in f^{-1}(z)\\
      \sum_{R: (x,y) \in R} w_{R} &\leq 1,&& \forall (x,y)\\
      w_{R} &\geq 0, &&\forall R\enspace .
    \end{align*}
\end{minipage}
\begin{minipage}{3in}
    \centerline{\underline{$\srec^{z,\mu}_{\eps,\delta}(f)$}}
    \begin{align}
     \min \quad  \sum_{R }  w_{R} \nonumber\\ 
       \sum_{(x,y) \in  f^{-1}(z)} \mu_{x,y}\sum_{R: (x,y)
               \in R} w_{R} &\geq (1 - \epsilon)\cdot \mu_z \label{eq:covering}\\
      \sum_{R: (x,y) \in R}
        w_{R} &\leq \delta, &&  \forall (x,y) \not\in f^{-1}(z) \label{eq:packing}\\
      \sum_{R: (x,y) \in R} w_{R} &\leq 1,&& \forall (x,y) \label{eq:1bd}\\
       w_{R} &\geq 0, && \forall R \enspace .\nonumber
    \end{align}
\end{minipage}
}

We will refer to the constraint in \eqref{eq:covering} as the
covering constraint and the ones in \eqref{eq:packing} as the
packing constraints. Note that while there is a single covering
constraint (averaged over all the inputs $(x,y)$ that satisfy $f(x,y)
=z$) there are packing constraints corresponding to each $(x,y) \notin f^{-1}(z)$.
\end{definition}
Similar to Yao's minmax principle~\cref{thm:yao}, we have the
following proposition relating the distributional version of the
smooth rectangle bound to the smooth rectangle bound. 

\begin{proposition}
\label{prop:yao}
$\srec_{\epsilon,\delta}(f) = \max_{\mu} \srec_{\epsilon,\delta}^{\mu}(f) $.
\end{proposition}

The main result of this section is the following
\begin{theorem}\label{thm:sreccc}
For any Boolean function $f: \{0,1\}^n\times\{0,1\}^n \to \{0,1\}$ and
any {\em product} distribution $\mu$ on $\{0,1\}^n\times\{0,1\}^n$, we
have the following.\changes
\begin{enumerate}
\item $$\CC^\mu_{0.49} (f) = O\left((\log \srec^{\mu}_{1/n^2,1/n^2}(f))^2 \cdot 
\log n\right).$$
\item Furthermore, if there exists $k \geq 20$ such that
\[
\lceil 100 \log
\srec^{\mu}_{\delta,\delta}(f) \rceil \leq k, 
\]
for $ \delta
\leq 1/(30\cdot 100 (k+1)^4)$, then $$\CC^{\mu}_{0.49}(f)  = O(k^2).$$
\end{enumerate}
\end{theorem}
The above theorem is useful only when we have a upper bound on the
smooth rectangle bound for very small $\delta$. The following
proposition shows that such upper bounds for smooth rectangle bound for
such small $\delta$ can be obtained in terms of either the information
complexity or the partition bound. 

\begin{proposition}\label{prop:sr_rprt}
For any Boolean function $f: \{0,1\}^n\times\{0,1\}^n \to \{0,1\}$ and
any $\delta \in (0,1)$, we have the following bounds on $\srec_{\delta,\delta}(f)$. 
\begin{align*}
\log \srec_{\delta,\delta}(f) &\leq O\left(\log \frac1{\delta}\right)
\cdot \IC_{1/8}(f),\\
\log \srec_{\delta,\delta}(f) &\leq O\left(\log \frac1{\delta}\right)
\cdot \log \prt_{1/8}(f).
\end{align*}
\end{proposition}
(This proposition depends on the error-reduction properties of
information complexity and partition bound; a proof appears in
\cref{sec:ccapp}.)  Using this proposition, we can reduce the
error (i.e., $\delta$) to $1/n^2$ and show that $\CC^\mu_{0.49}(f) =
O\left(\left(\log \prt_{1/8}(f)\right)^2 \cdot \left(\log
n\right)^3\right)$. However, we can also reduce the error to $1/\poly
(\log \prt_{1/8}(f))$ and show that there exists a $k =
O\left(\log \prt_{1/8}(f) \cdot \log \log \prt_{1/8}(f)\right)$ that
satisfies the hypothesis for the second part of \cref{thm:sreccc}.
The bound~\eqref{eq:prt} in \cref{thm:ccquad} now follows by combining
\cref{prop:sr_rprt,prop:yao,thm:sreccc}. A similar argument yields the
bound \eqref{eq:IC}. 

In particular, the above discussion shows that our techniques apply to
any complexity measure (not necessarily partition bound and
information complexity) which can be used to bound the smooth
rectangle bound for very small $\delta$. An interesting question that
arises in this context is if we could bound smooth rectangle bound for
small $\delta$ in terms of smooth rectangle bound for large $\delta$,
say $\delta = 1/3$ (i.e., is error-reduction for $\srec$
feasible?). This question was answered in the negative for partial
functions by G\"o\"os~et al.~\cite{GoosLMWZ2015} who show that there
exists a partial function $f$ that has $\srec_{1/3}(f) = O(\log n)$
and yet $\srec_{1/4}(f) = \Omega(n)$.

\subsection{Proof of {\cref{thm:sreccc}}}

In this section, we construct a communication protocol tree with a
small number of leaves from the optimal solutions to the LPs
corresponding to $\srec^{0,\mu}_{\eps,\delta}$ and
$\srec^{1,\mu}_{\eps,\delta}$. The construction of the protocol tree
with a small number of leaves is inspired by a construction due to
Nisan and Wigderson, in the context of log-rank
conjecture~\cite[Theorem 2]{NisanW1995} (see also \cite[Combinatorial
proof of Theorem~2.11]{KushilevitzNisan}). Unlike the earlier
constructions, our protocol works for a distribution and allows for
error. As a result, the decomposition into sub-problems needs to be
performed more carefully. This step critically uses the product nature
of the distribution $\mu$.


The decomposition is accomplished using an inductive argument. We will
work with the quantity $\srec^0 + \srec^1$. That is, we will show that if
this sum is small, then there is a protocol with few leaves.  Suppose
$\srec^0 \leq \srec^1$. Since $\srec^0$ is small, we will conclude
that there is a large rectangle biased towards $0$
(see Lemma~\ref{lem:bigrectangle}). Based on this large rectangle, the
entire communication matrix is partitioned into three parts: (1) the
large biased rectangle itself, (2) a rectangle whose corresponding
sub-problem admits an LP solution leading to a smaller $\srec^1$ value
(the underlying product nature of the distribution $\mu$ is 
used here) and (3) a rectangle where the total measure with respect to
$\mu$ drops significantly (see \cref{lem:subproblem}). 

We say that a rectangle $R$ is $(1-\alpha)$-biased towards to $0$ if 
$\mu_1(R) \leq \alpha \mu_0(R)$.
\begin{lemma}[large biased rectangle]\label{lem:bigrectangle}
Let $\mu$ be a product distribution.  If
$\srec^{0,\mu}_{\epsilon,\delta} (f)\leq D$, then for every $\rho \in
(0,1)$ there exists a rectangle $S$ such that $S$ is $(1-\rho)$-biased
towards $0$ and $$\mu(S) \geq \mu_0(S) \geq \frac{1}{D}\cdot \left(
(1-\epsilon)\cdot \mu_0
- \left(\frac\delta\rho\right)\cdot \mu_{1}\right).$$
\end{lemma}
(The proof appears in \cref{sec:deferredlemmas}.)
We will apply the above lemma with $\rho = \sqrt{\delta}$ and conclude
that there exists a large rectangle $S= X_0 \times Y_0$ that is
$(1-\sqrt{\delta})$-biased towards $0$. Let $X_1 = \X \setminus X_0$
and $Y_1 = \Y \setminus Y_0$. For $i,j \in \{0,1\}$, define rectangles
$R^{(ij)} := X_i \times Y_j$, $R^{(1*)} := X_1 \times \Y$, and
$R^{(*1)} := \X \times Y_1$. (Note, $S = R^{(00)}$.) For $i,j \in
\{0,1,*\}$, let $\mu^{(ij)}$ be
the restriction of $\mu$ to the rectangle $R^{(ij)}$. We show in the
lemma below that the function $f$ when restricted to either $R^{(10)}$
or $R^{(01)}$ has the property that the corresponding $\srec^1$ drops
by a constant factor. Define
\begin{align*}
\epsilon(f) &:= 1-\frac{\left(\sum_{(x,y) \in  f^{-1}(1)}
  \mu_{x,y}\sum_{R: (x,y)\in R} w_{R}\right)}{\mu_1}, \\
\epsilon^{(ij)}(f) &:= 1-\frac{\left(\sum_{(x,y) \in
                              f^{-1}(1)\cap R^{(ij)}}
  \mu_{x,y}\sum_{R: (x,y)\in R} w_{R}\right)}{\mu_1(R^{(ij)})};&
                                                                 \text{for }
                                                                 i, j
                                                                 \in \{0,1\}. 
\end{align*}
It follows from the covering constraint that $\epsilon(f) \leq
\epsilon$. Furthermore, $\epsilon(f)$ is an average of the $\epsilon^{(ij)}$'s
in the sense that $\epsilon(f) = \left(\sum_{i, j \in \{0,1\}}
\mu_1(R^{(ij)}) \epsilon^{(ij)}\right)/\mu_1$.

\begin{lemma}\label{lem:subproblem} Suppose the product
distribution $\mu$ and rectangles $R^{(ij)}$ are as above; in
particular, $R^{(00)}$ is $(1-\sqrt{\delta})$-biased towards $0$. There
exists an $(ij)\in \{(01),(10)\}$ such that one of the following
holds: (a) $2\mu^{(ij)}(f^{-1}(1)) \leq \mu^{(ij)}(f^{-1}(0))$ or (b)
$\srec^{1,\mu^{(ij)}}_{\epsilon^{(ij)}+30\sqrt[4]{\delta}, \delta}
(f) \leq 0.9 D$ where $\epsilon^{(ij)}$ is as defined above.
\end{lemma}
We will prove this lemma in \cref{sec:deferredlemmas}.  Let us assume
the above lemmas and obtain the low cost communication protocol
claimed in \cref{thm:sreccc}.

Suppose $\mu^{(01)}$ satisfies $\srec^{1,\mu^{(01)}}_{\eps^{(01)} +
  30\sqrt[4]{\delta},\delta}(f) \leq 0.9D$ as given by the above
  lemma. Consider the decomposition of the space
  $\X\times \Y$ given by $(R^{(00)}, R^{(01)}, R^{(1*)} =
  R^{(10)} \cup R^{(11)})$. We note that $R^{(00)}$ is a large biased
  rectangle, $R^{(01)}$ has lower $\srec^1$ value while $R^{(1*)}$ has
  lower $\mu$ value (since $R^{(00)}$ is large) and its $\srec$ values
  are no larger than that of the entire space. In the case when
  $\mu^{(10)}$ satisfies $\srec^{1,\mu^{(10)}}_{\eps^{(10)} +
  30\sqrt[4]{\delta},\delta}(f) \leq 0.9D$, we similarly have the
  decomposition $(R^{(00)}, R^{(10)}, R^{(*1)} = R^{(01)} \cup
  R^{(11)})$.

 This suggests a natural inductive protocol $\Pi$ for $f$ that we
formalize in the lemma below.

For our induction it will be convenient to work with $\mu$
that are not necessarily normalized. So, we will only assume
$\mu: \X \times \Y \rightarrow [0,1]$ but not that $|\mu| := \mu(\X
\times \Y) =\sum_{(x,y) \in \X \times \Y}\mu(x,y)= 1$. For a protocol $\Pi$, let the advantage of $\Pi$ be defined by
\[ \adv_\mu(\Pi) = \sum_{(x,y): f(x,y) = \Pi(x,y)}  \mu(x,y)
 - \sum_{(x,y): f(x,y) \neq \Pi(x,y)} \mu(x,y). \] Let $L(\Pi)$ be the
number of leaves in $\Pi$.

%
We now formulate the induction hypothesis as follows.
\begin{lemma}\label{lem:induction}
Fix a function $f: \X \times \Y \rightarrow \{0,1\}$ and a product
distribution (not necessarily normalized)
$\mu: \X \times \Y \rightarrow [0,1]$ such that $|\mu| \geq 0$.
Let $\ve, \delta \in (0,1)$ and $\Delta \in (0,|\mu|)$.
Let $s, t$ be non-negative integers such that 
\begin{align*}
s \geq s(\mu,\epsilon,\delta) &:= \ceil{100 \cdot \log 
2(\srec^{0,\mu}_{\epsilon,\delta}(f) + \srec^{1,\mu}_{\epsilon,\delta}(f))};\\
t \geq t(\mu,\epsilon,\delta) &:= \ceil{100 \cdot 2^s\log (|\mu|/\Delta)}.
\end{align*}
Then, there is a protocol $\Pi$ such that 
\begin{align}
L(\Pi) & \leq 4{s + t \choose t} -1 ;\\
\adv_\mu(\Pi)& \geq \left(\frac{1}{10}- \epsilon - 30(s+1)
               \sqrt[4]{\delta}\right)|\mu| - \Delta \cdot L(\Pi).
\label{advbound}
\end{align}
\end{lemma}
\begin{remark}
Since $\ve \leq \frac{1}{2}$, our definitions imply that
$\srec^{1,\mu}_{\epsilon,\delta}(f)
+ \srec^{1,\mu}_{\epsilon,\delta}(f)) \geq \frac{1}{2}$; thus $s \geq 0$.
Similarly, since $\Delta \leq |\mu|$, we have $t \geq 0$.
\end{remark}

\begin{proof}
  First, we observe that if $\max\{\mu_0,\mu_1\} \geq
  2 \min\{\mu_0,\mu_1\}$, then the protocol $\Pi$ consisting of just
  one leaf, with the most popular value as label, meets the
  requirements: for, $\adv_\mu(\Pi) \geq \frac{1}{3}|\mu|$ and $L(\Pi)
  =1$, and our claim holds. Also, we may assume that $\epsilon -
  30(s+1) \sqrt[4]{\delta} < \frac{1}{10}$, for otherwise the claim is
  trivially true.

  We now proceed by induction on $s+t$, assuming
  that $\mu$ is balanced: $\max\{\mu_0,\mu_1\} \leq
  2 \min\{\mu_0,\mu_1\}$.

  \paragraph*{Base case $(s=0)$:} 
  Since $s=0$, we have
  $\log \srec^{1,\mu}_{\epsilon,\delta}(f) \leq \frac1{100}$. 
  We will show a protocol $\Pi$ where Alice sends one bit after which
  Bob announces the answer. Consider the optimal solution
  $\langle w_R: \mbox{$R$ a rectangle}\rangle$ to the LP corresponding
  to $\srec^{1,\mu}_{\epsilon,\delta}(f)$; thus,
  $\OPT := \sum_R w_R = \srec^{1,\mu}_{\epsilon,\delta}(f) \leq
  2^{1/100} \leq 2$.
  Let $R = R_X \times R_Y$ be a random rectangle picked with
  probability proportional to $w_R$ (using public coins).  In the
  protocol $\Pi$, Alice tells Bob if $x \in R_X$, and Bob returns the
  answer $1$ if $(x,y) \in R_Y$ and returns $0$ otherwise. Let
  $p_{xy} := \Pr_R[(x,y) \in R]$. Then, by \eqref{eq:covering} we have 
  $\sum_{(x,y) \in f^{-1}(1)} \mu(x,y)p_{xy} \geq (1 - \epsilon)
  \mu_1/\OPT$,
  and by \eqref{eq:packing}, we have
  $\sum_{(x,y) \in f^{-1}(0)} \mu(x,y) p_{xy} \leq \delta
  \mu_0/\OPT$. Thus,
\begin{align}
\E_R\left[\sum_{(x,y):\Pi(x,y) \neq f(x,y)} \mu(x,y)\right]  
&= \sum_{(x,y)\in f^{-1}(1)} \mu(x,y) (1-p_{xy}) +
                        \sum_{(x,y)\in f^{-1}(0)} \mu(x,y) p_{x,y} \nonumber\\
&\leq  \mu_1- (1-\epsilon)\mu_1/\OPT + \delta \mu_0/\OPT \nonumber \\
&\leq  \mu_1- ((1-\epsilon)\mu_1- \delta \mu_0) / \OPT \nonumber\\
&\leq  \frac{1}{2}(\mu_1+\epsilon \mu_1 + \delta \mu_0)\qquad (\mbox{since $\OPT \leq 2$}) \label{eq:expectation}.
\end{align}
Fix a choice $R$ for which the quantity under the expectation is at
most $\frac{1}{2}(\mu_1+\epsilon \mu_1 + \delta \mu_0)$. Then,
\begin{align*}
\adv(\Pi) &= |\mu| - 2 \sum_{(x,y):\Pi(x,y) \neq f(x,y)} \mu(x,y)\\
           &\geq |\mu| - (\mu_1 + \epsilon \mu_1 + \delta \mu_0)\\
           &\geq  \left(\frac{1}{3} - \epsilon - \delta\right) |\mu| & \mbox{(since $\mu_1 \leq 2 \mu_0$)}.
\end{align*}

\paragraph*{Base case $(t=0)$:} In this case, $|\mu| = \Delta$, and the 
protocol $\Pi$ with a single leaf that gives the most probable answer
achieves $\adv(\Pi) \geq 0 \geq |\mu| - \Delta$.

\paragraph*{Induction step:} We will use \cref{lem:bigrectangle} 
to decompose the communication matrix into a small number of
rectangles. After an exchange of a few bits to determine in which
rectangle the input lies, Alice and Bob will be left with a problem
for which $s$ or $t$ is significantly smaller.  Assume
$\srec_{\epsilon,\delta}^{1,\mu}(f) \geq 
\srec_{\epsilon, \delta}^{1,\mu}(f)$; 
in particular, $\srec_{\epsilon, \delta}^{1,\mu}(f) \leq 2^{s/100}$.

Formally, from \cref{lem:bigrectangle} (taking $\rho
= \sqrt{\delta}$), we obtain a rectangle $R^{(00)} = X_0 \times Y_0$
such that (a) $R^{(00)}$ is $(1-\sqrt{\delta})$-biased towards $0$,
and (b) $\mu(R^{(00)}) \geq \frac{1}{2^{s/100}} (1-\epsilon -
2\sqrt{\delta}) |\mu_0| \geq \frac{1}{3 \cdot 2^{s/100}} (1-\epsilon -
2\sqrt{\delta}) |\mu|$.  Recall the definitions of the rectangles
$R^{(10)}, R^{(01)}, R^{(11)}, R^{(1*)}, R^{(*1)}$ and the
corresponding restrictions of $\mu$, namely,
$\mu^{(01)}, \mu^{(10)}, \mu^{(11)}, \mu^{(1*)}, \mu^{(*1)}$. Suppose
the choice of $ij$ in \cref{lem:subproblem} for which one of the
alternatives holds is $ij =01$ (the other case $ij=10$ is symmetric).
The protocol $\Pi$ proceeds as follows. Alice starts by telling Bob if
$x \in X_0$.
\begin{description}
\item[Alice says $x \in X_0$.] Now, Bob tells Alice if $y \in Y_0$. 
\begin{description}
\item[Bob says $y \in Y_0$.] The protocol $\Pi^{(00)}$ in this case has 
one leaf with answer $0$; thus $\adv(\Pi^{(00}) \geq
|\mu^{(00)}| \cdot (1-\sqrt{\delta})$.

\item[Bob says $y \not \in Y_0$.] Alice and Bob follow the
protocol $\Pi^{(01)}$ promised by induction for $R^{(01)}$ under
$\mu^{(01)}$. To bound the number of leaves in $\Pi^{(01)}$, we will
consider the two alternatives ((a) and (b)) specified
in \cref{lem:subproblem} separately. First (alternative (a)) suppose
$2\mu^{(01)}(f^{-1}(1)) \leq \mu^{(01)}(f^{-1}(0))$; then we immediately
declare $0$ as the response, so that $L(\Pi^{(01)}) = 1$ and
$\adv(\Pi^{(01)}) \geq |\mu^{(01)}|/3$. If alternative (b)
holds, then we have
\begin{equation}
 \srec^{1,\mu^{(01)}}_{\epsilon^{(01)} + 30\sqrt[4]{\delta}, \delta}(f) \leq
   0.9 \srec^{1,\mu}_{\epsilon, \delta}(f). \label{eq:sreduction}
\end{equation} 
Then, we obtain $\Pi^{(01)}$ by induction. We take $\epsilon^{(01)} +
30\sqrt[4]{\delta}$ as $\epsilon$ (if this quantity is greater than
$1$, then we use a trivial protocol with one leaf and zero
advantage). With the reduction promised in \eqref{eq:sreduction}, we
may use a value of $s$ that is the old $s$ minus 1. Thus, we have
\begin{align*}
L(\Pi^{(01)}) & \leq 4{(s-1) + t \choose t} -1;\\
\adv(\Pi^{(01)}) & \geq 
      |\mu^{(01)}| \cdot\left (\frac{1}{10} - (\epsilon^{(01)} +
            30\sqrt[4]{\delta})- 30s\sqrt[4]{\delta}\right) - \Delta \cdot L(\Pi^{(01)}).
\end{align*}

\end{description}
\item[Alice says $x \not\in X_0$.]
Alice and Bob follow the protocol $\Pi^{(1*)}$ obtained by applying
the induction hypothesis to the rectangle $R^{(1*)}$ and the
associated distribution $\mu^{(1*)}$.  Observe that
\begin{equation}|\mu^{(1*)}| \leq |\mu| - \mu(R^{(00)}) \leq |\mu| \left(1-\frac{1}{3\cdot 2^{s/100}} (1-\epsilon -
2\sqrt{\delta})\right) \leq |\mu| \left(1-\frac1{4\cdot 2^{s}}\right).
\label{eq:mureduction}
\end{equation}  
For the last inequality we used $\epsilon +
2 \sqrt{\delta} \leq \frac{1}{10}$, for otherwise \cref{advbound}
holds trivially. Now, \cref{eq:mureduction} implies that 
$\log |\mu^{(1*)}| \leq \log |\mu| - \frac{1}{1002^s}$; so,
for our induction we may take $t \leftarrow t-1$. The parameters
$\ve$, $\delta$ and $\Delta$ remain the same. The original LP
solutions are still valid for the subproblem, so we use the same
$s$. The protocol $\Pi^{(1*)}$ obtained by induction satisfies the
following inequalities.
\begin{align*}
L(\Pi^{(1*)}) &\leq 4{s + (t-1) \choose t-1}  - 1;\\
\adv(\Pi^{(1*)}) &\geq |\mu^{(1*)}| \cdot \left(\frac{1}{10} - \epsilon^{(1*)} - 30(s+1)\sqrt[4]{\delta}\right)
        - \Delta \cdot L(\Pi^{(1*)}).
\end{align*}
\end{description}
Putting all the contributions together, we obtain
\begin{align*}
L({\Pi}) & =1  + L(\Pi^{(01)})+ L(\Pi^{(1*)})\\
&\leq 1  +  \left(4{(s-1) + t \choose t} - 1 \right) + \left(4{s + (t-1) \choose t-1}  - 1\right)\\
        &=  4{s + t \choose t} - 1;\\
\adv(\Pi) & \geq |\mu^{(00)}| \cdot (1-\sqrt{\delta}) \\
          &   +
      |\mu^{(01)}| \cdot\left (\frac{1}{10} - (\epsilon^{(01)} +
            30\sqrt[4]{\delta})- 30s\sqrt[4]{\delta}\right) - \Delta \cdot L(\Pi^{(01)})\\
      & +  |\mu^{(1*)}| \cdot \left(\frac{1}{10} - \epsilon^{(1*)} - 30(s+1)\sqrt[4]{\delta}\right)
        - \Delta \cdot L(\Pi^{(1*)}) \\
      &\geq \left(\frac{1}{10} - \epsilon - 30(s+1)
        \sqrt[4]{\delta}\right)|\mu| - \Delta \cdot L(\Pi).     \qedhere
\end{align*}
\end{proof}
The above lemma yields a protocol whose
protocol tree has a small number of leaves, but not necessarily small
depth. We can balance the protocol tree using the following proposition.
\begin{proposition}[{\cite[Lemma~2.8]{KushilevitzNisan}}]\label{prop:balancing}
If $f$ has a deterministic communication protocol
tree with $\ell$ leaves, then $f$ has a protocol tree with depth at most
$O(\log \ell)$.
\end{proposition}

We are now in a position to complete the proof of the main
theorem of this section.
\begin{proof}[Proof of {\cref{thm:sreccc}}]
To prove the first part of \cref{thm:sreccc}, we invoke \cref{lem:induction} with $\Delta = 1/2^{4n}$ and $\epsilon
= \delta = 1/n^2$ to derive a protocol tree $\Pi$ with at
most $$L(\Pi) = n^{{O\left(\log
      \srec^{1,\mu}_{1/n^2,1/n^2}(f)\right)}^2}$$ leaves and advantage at least
$1/20$. The first part now follows from \cref{prop:balancing}.

To prove the second part of \cref{thm:sreccc}, we
invoke \cref{lem:induction} with $s=k$, $\Delta = 1/2^{5k^2}$ and
$\epsilon = \delta = 1/(30\cdot 100(k+1)^4)$ where $k$ satisfies the
hypothesis. With this setting of parameters $t = \lceil 500\cdot 2^k
k^2
\rceil \leq 2^{2k}$ (for $k \geq 20$).  \cref{lem:induction} implies a
protocol tree $\Pi$ with at most
$$L(\Pi) \leq (t+s)^s \leq t^{2s} \leq 2^{4k^2}$$ leaves and advantage
at most $1/20$. The second claim then follows from \cref{prop:balancing}.
\end{proof}


\subsection{Proofs of {\cref{lem:bigrectangle,lem:subproblem}}}
\label{sec:deferredlemmas}

\begin{proof}[Proof of {\cref{lem:bigrectangle}}]
Fix $z \in \{0,1\}$. In the following we say that 
a rectangle $R$ is {\em biased} (towards $0$) if $\mu_{1}(R)
\leq \rho \cdot \mu_0(R)$; otherwise, we say it is unbiased. Fix a
solution $\langle w_R :  
R \text{ is a rectangle} \rangle$ that achieves the optimum $\srec_{\epsilon, \delta}^{0,\mu}(f) \leq D$. It
follows
\begin{align*}
\sum_{R:\text{unbiased}}  w_R \cdot \mu_0(R) &\leq \sum_{R: \text{unbiased}}
  w_R \cdot \frac{\mu_{1}(R)}{\rho} \\
&\leq \frac1\rho\cdot \sum_R w_R\cdot \mu_{1}(R)\\
&= \frac1\rho\sum_{(x,y) \in f^{-1}(1)} \mu(x,y) \sum_{R: (x,y)
  \in R} w_R \\
&\leq \frac{\delta}{\rho}\cdot \mu_{1},
\end{align*}
where the last inequality follows from the packing constraints
\eqref{eq:packing}.
We now use the covering constraints \eqref{eq:covering} to conclude that
\begin{align}
\sum_{R: \text{biased}} w_R \cdot \mu_0(R) = \sum_R w_R \cdot \mu_0(R) - \sum_{R:
  \text{unbiased}} w_r \cdot \mu_0(R) \geq (1-\epsilon) \cdot \mu_0 -
  \frac\delta\rho\cdot \mu_{1}.\label{eq:biaslarge}
\end{align}
Define a probability distribution on the rectangles $R$ as follows
$p(R) : =
w_R/\srec_{\epsilon,\delta}^{0,\mu}(f)$. Then \eqref{eq:biaslarge} can
be rewritten as
\begin{align*}
\E_R\left[ \I_{\text{biased}}(R) \cdot \mu_0(R) \right] \geq \frac1D
  \cdot \left( (1-\epsilon)\cdot \mu_0 - \frac\delta\rho\cdot \mu_{1}\right).
\end{align*}
Hence, there exists a large biased rectangle $S=X_0\times Y_0$ as claimed.
\end{proof}

\begin{proof}[Proof of {\cref{lem:subproblem}}]
Since $R^{(00)}$ is $(1-\sqrt{\delta})$-biased towards 0, we have from
the packing and covering constraints \eqref{eq:packing}
and \eqref{eq:1bd} that
\begin{align*}
\sum_{(x,y) \in R^{(00)}} \mu_{x,y} \sum_{R \ni (x,y)} w_R & =
                                                          \sum_{(x,y)
                                                          \in R^{(00)} \cap
                                                          f^{-1}(1)}
                                                          \mu_{x,y}
                                                          \sum_{R \ni
                                                          (x,y)} w_R + \sum_{(x,y)
                                                          \in R^{(00)} \cap
                                                          f^{-1}(0)}
                                                          \mu_{x,y}
                                                          \sum_{R \ni
                                                          (x,y)} w_R\\
&\leq \mu_1(R^{(00)}) + \delta \mu_0(R^{(00)}) \leq (\sqrt{\delta} + \delta)\mu_0(R^{(00)}) \leq
  2\sqrt{\delta} \mu(R^{(00)}).
\end{align*}
Hence, 
\begin{equation}
\sum_R w_R \left(\frac{\mu(R^{(00)} \cap R)}{\mu(R^{(00)}) }\right)  \leq                    2\sqrt{\delta}.\label{eq:SRbound} 
\end{equation} 
Group the rectangles in to subsets as follows:
\begin{align*}
B^{(01)} &:= \left\{ R : \frac{\mu(R^{(01)}\cap R)}{\mu(R^{(01)})} \geq
    \frac{10\sqrt[4]{\delta}}{D}\right\},&
B^{(10)} &:= \left\{ R : \frac{\mu(R^{(10)}\cap R)}{\mu(R^{(10)})} \geq
    \frac{10\sqrt[4]{\delta}}{D}\right\},\\
B &:= \left\{ R : \frac{\mu(R^{(11)}\cap R)}{\mu(R^{(11)})} \geq
    \frac{10}{D}\right\}.
\end{align*}

By \eqref{eq:1bd}, we have
\begin{align*}
\sum_{(x,y) \in R^{(11)}} \mu_{x,y} \sum_{R \ni (x,y)} w_R \leq \sum_{(x,y)
  \in R^{(11)}} \mu_{x,y} = \mu(R^{(11)}).
\end{align*}
Or equivalently,
\begin{align*}
\sum_R \frac{w_R}{D} \cdot \frac {\mu(R^{(11)} \cap R)}{\mu(R^{(11)})} \leq
  \frac1{D}.
\end{align*}
Hence,
\begin{equation}
\sum_{R \in B} w_R \leq 0.1 D. \label{eq:Bbound}
\end{equation}
We will now argue that either $\sum_{R \in B^{(01)}} w_R \leq 0.9 D$
or $\sum_{R \in B^{(10)}} w_R \leq 0.9 D$. Suppose, for contradiction,
that neither is true. Then, by \eqref{eq:Bbound} we have
\begin{align}
\sum_{R \in B^{(01)} \cap B^{(10)} \cap \Bar{B}} w_R \geq 0.7 D.\label{eq:B1B2Bbarbd}
\end{align}
Since $\mu$ is a product distribution we have
\begin{align*} %
\frac{\mu(R^{(01)} \cap R)}{\mu (R^{(01)})} \cdot \frac{\mu(R^{(10)} \cap R)}{\mu
  (R^{(10)})} = \frac{\mu(R^{(00)}\cap R)}{\mu (R^{(00)})} \cdot \frac{\mu(R^{(11)}\cap R)}{\mu (R^{(11)})}.
\end{align*}
Using the above we have
\begin{align*}
\sum_{R} w_R \left(\frac{\mu(R^{(00)} \cap R)}{\mu(R^{(00)}) }\right) 
&\geq \sum_{R\in B^{(01)} \cap B^{(10)} \cap \Bar{B}} w_R \left(\frac{\mu(R^{(00)} \cap R)}{\mu(R^{(00)}) }\right) 
\\
& \geq \left. \sum_{R\in B^{(01)} \cap B^{(10)} \cap \Bar{B}} w_R 
\left(\frac{\mu(R^{(01)} \cap R)}{\mu (R^{(01)})}\right) \cdot\left(\frac{\mu(R^{(10)} \cap R)}{\mu  (R^{(10)})}\right) \middle/ \left(
\frac{\mu(R^{(11)}\cap R)}{\mu (R^{(11)})}\right) \right.
\\
&\geq \left. \sum_{R\in B^{(01)} \cap B^{(10)} \cap \Bar{B}} w_R 
  \left(\frac{10\sqrt[4]{\delta}}{D}\right) \cdot
  \left(\frac{10\sqrt[4]{\delta}}{D}\right)  \middle/ \left(\frac{10}{D}\right) \right.\\
& \geq \frac{10\sqrt{\delta}}{D} \cdot (0.7 D)\\
& =  7\sqrt{\delta}.
\end{align*} %
This contradicts \eqref{eq:SRbound}. Hence, either $\sum_{R \in
B^{(01)}} w_R \leq 0.9 D$ or $\sum_{R \in B^{(10)}} w_R \leq 0.9
D$. Assume, wlog that $\sum_{R \in B^{(01)}} w_R \leq 0.9 D$. If $f$
is $1/2$-biased towards 0 with respect to the distribution
$\mu^{(01)}$, then the alternative (a) of the lemma holds, and we are
done.  Otherwise, that is $\mu_0(R^{(01)}) \leq 2\mu_1(R^{(01)})$ or
equivalently $\mu(R^{(01)}) \leq 3\mu_1^{(01)}(R^{(01)})$. We will
infer from this that $\srec^{1,\mu^{(01)}}_{\epsilon^{(01)}
+30\sqrt[4]{\delta},\delta}(f) \leq 0.9 D$.  Consider the primal
solution given by
\begin{align*}
w'_R = \begin{cases}
  w_R, & \text{if } R \in B^{(01)}\\
0, & \text{if } R \notin B^{(01)}.
\end{cases}
\end{align*}
Clearly, $w'_R$, being a part of the original solution,
satisfies \eqref{eq:packing} and
\eqref{eq:1bd}, and has objective value at most $0.9 D$. All we need
to show is that it satisfies the covering constraint~\eqref{eq:covering}. For this, we first
consider
\begin{align}
\sum_{R \in \Bar{B^{(01)}}} w_R \left( \frac{\mu_1(R^{(01)} \cap R)}{\mu(R^{(01)})}\right)
  \leq \sum_{R \in \Bar{B^{(01)}}} w_R \left( \frac{\mu(R^{(01)} \cap
  R)}{\mu(R^{(01)})} \right)\leq \frac{10\sqrt[4]{\delta}}{D}\cdot D \leq 10\sqrt[4]{\delta}. \label{eq:bbarbd}
\end{align}
Now,
\begin{align*}
&\sum_{(x,y) \in f^{-1}(1) \cap R^{(01)}} \mu_{x,y} \sum_{R \in (x,y)} w'_R \\
& = \sum_{(x,y) \in f^{-1}(1) \cap R^{(01)}} \mu_{x,y} \sum_{R \in (x,y), R
  \in B^{(01)}} w_R \\
& = \sum_{(x,y) \in f^{-1}(1) \cap R^{(01)}} \mu_{x,y} \left( \sum_{R \in
  (x,y)} w_R - \sum_{R \in (x,y) , R \notin B^{(01)}} w_R \right)\\
& =  (1-\epsilon^{(01)}) \mu_1(R^{(01)}) - \sum_{(x,y) \in f^{-1}(1) \cap R^{(01)}}
  \mu_{x,y} \sum_{R \in (x,y) , R \notin B^{(01)}} w_R \\
& = (1-\epsilon^{(01)}) \mu_1(R^{(01)}) - \sum_{R \notin B^{(01)}} w_R \mu_1(R^{(01)}
  \cap R)\\
& \geq (1-\epsilon^{(01)}) \mu_1(R^{(01)}) - 10\sqrt[4]{\delta} \mu(R^{(01)}) &
                                                             \text{[From
                                                             \eqref{eq:bbarbd}]}\\ 
& \geq (1-\epsilon^{(01)}) \mu_1(R^{(01)}) - 30\sqrt[4]{\delta} \mu_1(R^{(01)})&
                                                                  \text{[Since
                                                                  $\mu(R^{(01)})
  \leq 3\mu_1(R^{(01)})$]}\\
& = (1-\epsilon^{(01)} -30\sqrt[4]{\delta}) \mu_1(R^{(01)})
\end{align*}
Thus, \eqref{eq:covering} holds for $R^{(01)}$ with $\epsilon$
replaced by $\epsilon^{(01)} + 30\sqrt[4]{\delta}$.
\end{proof}

\section{Query Complexity}
\label{sec:q}
Let $f:\{0,1\}^n \rightarrow \{0,1\}$ be the function
for which we wish to build a decision tree.  
\subsection{Preliminaries}
\begin{definition}[product distribution]
We say $\mu: \{0,1\}^n \rightarrow [0,1]$ is a (bit-wise) product distribution
on $\{0,1\}^n$ if for $i=1,2, \ldots n$, there exist $p_i(0),p_i(1) \in [0,1]$ 
(satisfying $p_i(0) + p_i(1)=1$) such that for all $x \in \{0,1\}^n$,
$\mu(x)=\prod_i p_i(x_i)$.
\end{definition}
Let $\mu$ be a bit-wise product distribution on the inputs to $f$.
Our goal is to build an efficient decision tree $\cT$ for $f$ such
that $\Pr_\mu[f(x) \neq \cT(x)]$ is small. As the input bits are
queried, the input is restricted to reside in a subcube. Let
$s \in \{0,1,\star\}^n$. The subcube of $\{0,1\}^n$ with support $s$
is
\[ \subcube(s):= \{x \in \{0,1\}^n: s_i \neq \star 
                                    \Rightarrow x_i =s_i\}.\] 
Its size
is $\size(s) := |\{i: s_i \in \{0,1\}\}|.$ If $A$ is a subcube, say
$A=\subcube(s)$, then  $|A|:=\size(s)$ .

We will derive an efficient decision tree in terms of the {\em query
partition bound} of $f$ due to Jain and Klauck~\cite{JainK2010}. The
relationship of the query partition bound with other LP based bounds is
described in \cref{sec:qcapp}.
\begin{definition}[query partition bound~\cite{JainK2010}]
Let $\ve > 0$.  In the following, $A$ represents a subcube of
$\{0,1\}^n$.  The $\ve$-query partition bound of
$f:\{0,1\}^n \rightarrow \{0,1\}$, denoted $\qprt_\ve(f)$, is the
optimal value of the following linear program.
\begin{align*}
  {\min}\quad  \sum_{z} \sum_{A}  w_{z,A} \cdot 2^{|A|} \\
      \sum_{A: x \in A} w_{f(x),A} &\geq 1 - \epsilon,
   && \forall x \in \{0,1\}^n\\
     \sum_{A: x \in A} \quad \sum_{z}  w_{z,A} &= 1,  && \forall x\in \{0,1\}^n  \\
      w_{z,A} &\geq 0, && \forall (z,A) \enspace\\
         \end{align*}
\end{definition}

 
Using the following claim (see appendix) one can ensure that the error
parameter $\ve$ above is small with a modest increase in the query
partition bound.
\begin{claim}
Let $\ve > \delta >0$. Then $\log \qprt_\delta(f) \leq \left(\frac{2}{(0.5 - \ve)^2}\log\frac{1}{\delta} \right) \log \qprt_\ve(f) .$
\label{claim:boosting}
\end{claim}

\begin{definition}
 We say $\mu$ is an $(\alpha_0, \beta_0, \alpha_1,\beta_1, a,
 b)$-feasible distribution for $f$ if there exists a feasible solution to the
following inequalities. The variables are 
$$(u_R: R \text{ is a subcube with support of size at most } a)$$ 
and   
$$(w_R: R \text{ is a subcube
with support of size at most } b). $$

{
  \begin{minipage}{2.5in}
    \begin{align}
      \sum_{R: x \in R} u_R   & \geq 1-\alpha_0 & \forall x \in f^{-1}(0) \label{eq:zero1}\\
      \sum_{R: x \in R} u_R & \leq \beta_0  & \forall x \in f^{-1}(1) ;\label{eq:zero2}\\
      u_R &\geq 0  . \label{eq:zero3}
    \end{align}
  \end{minipage}\hspace{0.5in}
  \begin{minipage}{3.5in}
    \begin{align}
      \sum_R \mu_1(R)w_R   &\geq (1-\alpha_1) \mu_1 \label{eq:one1}\\
      \sum_{R: x \in R} w_R &\leq 1 & \forall x \in \{0,1\}^n; \label{eq:one1.5}\\
      \sum_{R: x \in R} w_R & \leq \beta_1  & \forall x \in f^{-1}(0) ; \label{eq:one2}\\
      w_R &\geq 0.
    \end{align}
  \end{minipage}
}

\end{definition}
\begin{remark}
If the $i$-th bit of the input is fixed in $\mu$ (that is,
$p_i(0) \in \{0,1\}$), then we will assume $i$ is not part of the
support of any $R$ in the above feasible solution.
\end{remark}

The main technical contribution of this section is the following.
\begin{lemma}
Let $\delta>0$ and let $\mu$ be a product distribution that is $(\alpha_0, \beta_0, \alpha_1,\beta_1, a,
b)$-feasible for $f$. Then, there is a decision tree for
$f$ of depth at most $ab$ that errs with probability at most
$$ \frac{1}{4} +  \alpha_1 + \beta_1 + 4b(\beta_1  + \delta) +
\frac{\beta_0}{(1-\alpha_0)\delta}.$$
\label{claim:main}
\end{lemma}

The query complexity bound stated in the introduction
(\cref{thm:qcquad}) follows from the above lemma.
\begin{proof}[Proof of {\cref{thm:qcquad}}]
Let $c :=8 + \log \qprt_\ve(f)$. Let $\gamma = 1/c^8$. Then from
\cref{claim:boosting} we get that $d :=\log \qprt_\gamma(f) = O(c
\log c)$. Let $\{ w_{z,A}\}$ be an optimal solution for the primal of
$\qprt_\delta(f)$. Let $\B :=\{ A ~:~ |A| > d + \log \frac{1}{\gamma}
\}$. Then $\sum_z \sum_{A \in \B} w_{z,A} < \gamma$ since $\sum_z
\sum_{A} w_{z,A} \cdot 2^{|A|} = 2^d$ . This implies (by first
boosting and then removing the $A \in \B$) that $\mu$ is an
$(\alpha_0, \beta_0, \alpha_1,\beta_1, a, b)$-feasible distribution for $f$, with $\alpha_0 = \beta_0 = \alpha_1 = \beta_1 = 2\gamma$ and $a = b = O(c \log c)$.

From \cref{claim:main} (by setting $\delta = 1/c^4$) we get that there
is a decision tree for $f$ of depth at most $O(c^2 \log^2 c)$ with
error under $\mu$ at most $0.49$.
\end{proof}

\subsection{Proof of {\cref{claim:main}}}

In this section, we show that if a product distribution $\mu$ is
feasible, then the functions admit a decision tree of low
complexity. This decision tree is obtained from the feasible solution
of the LP as follows. We first show that feasibility implies the
existence of a biased subcube of small support (see
\cref{cl:biasedcube}). After querying the support of this subcube, one
is left with several subproblems. One of the subproblems corresponds
to the subcube itself, in which case we answer according to its
bias. For each of the other subproblems, we observe that the induced
distribution $\mu$ admits a feasible solution consisting of rectangles
with a strictly smaller support size. This is proved by showing that
the contribution of rectangles whose supports are disjoint to the
original subcube is negligible (see \cref{cl:eliminate}). This step crucially uses the product
nature of the distribution $\mu$.
 
For a set $A \subseteq \{0,1\}^n$, let $\mu_0(A) = \mu(A \cap f^{-1}(0))$ and
$\mu_1(A) = \mu(A \cap f^{-1}(1))$; let $\mu_0 = \mu_0(\{0,1\}^n)$ and
$\mu_1 = \mu_1(\{0,1\}^n)$.

\begin{claim} 
\label{cl:biasedcube}
Suppose $\mu: \{0,1\}^n \rightarrow [0,1]$ is a product probability
distribution satisfying \eqref{eq:zero1} , \eqref{eq:zero2} and \eqref{eq:zero3}. Further, suppose $\delta >0$ is such that
\begin{equation}
(1 - \alpha_0) \mu_0 - \left(\frac{\beta_0}{\delta}\right)\mu_1 > 0. \label{eq:assumption}
\end{equation}
Then, there is subcube $A$ with support of size at most $a$, such that 
$\mu_1(A) \leq \delta \mu_0(1)$.
\end{claim}
\begin{proof}
We say $A$ is biased if $\mu_1(A) \leq \delta \mu_0(A)$.
From \eqref{eq:zero2}, we have
\begin{align}
\beta_0 \mu_1 &\geq \sum_{A} \mu_1(A) u_A \\
 & \geq \sum_{\mbox{\scriptsize $A$ not biased}} \mu_1(A)u_A \\ 
 & \geq \delta \sum_{\mbox{\scriptsize $A$ not biased}}  \mu_0(A) u_A.
\end{align}
Combining this with \eqref{eq:zero1}, we obtain
\begin{align}
(1-\alpha_0)\mu_0
- \left(\frac{\beta_0}{\delta}\right) \mu_1 & \leq
\sum_{A} u_A \mu_0(A) - \sum_{\mbox{\scriptsize $A$ unbiased}} u_A \mu_0(A)\\
&\leq  \sum_{\mbox{\scriptsize $A$ biased}} u_A \mu_0(A).
\end{align}
Since the left hand side is positive, the sum on the right cannot be
empty. The claim follows from this.
\end{proof}

\newcommand{\cB}{{\cal B}}
\begin{claim} \label{cl:eliminate}Let $\delta > 0$. Fix a product distribution $\mu$, and let $A$ be a subcube such that
$\mu_1(A) \leq \delta \mu_0(A)$.  Suppose $(w_R: R \mbox{ a
  subcube})$ satisfies \eqref{eq:one1.5} and \eqref{eq:one2}.
Let $\cB = \{B: \support(B) \cap \support(A) = \emptyset\}$.
Then,
$$ \sum_{B \in \cB} \mu_1(B) w_B  \leq \beta_1 + \delta . $$

\end{claim}
\begin{proof} We will use the product nature of the distribution in the following form: if $A$ and $B$ are subcubes with disjoint supports, then 
$\mu(A)\mu(B) = \mu(A \cap B)$. Thus, we have
\begin{align}
\sum_{B \in \cB} \mu_1(B) w_B&\leq \sum_{B \in \cB} \mu(B) w_B \\
                   &\leq \frac{1}{\mu(A)} \sum_{B \in \cB} \mu(A \cap  B) w_B \\
                   &\leq \frac{1}{\mu(A)} \sum_{B \in \cB}  \mu_0(A \cap
                    B)w_B + \frac{1}{\mu(A)} \sum_{B \in \cB} \mu_1(A \cap B) w_B.
\label{eq:twoterms}
\end{align}
Let us bound the two terms on the right separately.
For the first term, by our assumption \eqref{eq:one2}, we have
$$ \sum_{B \in \cB} \mu_0(A \cap B)w_B \leq \beta_1 \mu_0(A) \leq \beta_1 \mu(A). $$
For second term, by assumption \eqref{eq:one1.5}, we have
$$\sum_{B \in \cB} \mu_1(A \cap B) w_B \leq \mu_1(A) 
 \leq \delta \mu_0(A)  \leq \delta \mu(A). $$
By using these bounds in \eqref{eq:twoterms}, we establish our claim.
\end{proof}

%

\begin{proof}[Proof of {\cref{claim:main}}]
 We assume that $ \mu_0, \mu_1 \geq \frac{1}{4},$  for otherwise, we can reliably guess the answer without making any
query.  Similarly, we assume that \eqref{eq:assumption} holds, for
otherwise, we may answer 1 without making any query, and yet err with 
probability at most
$$ \mu_0 \leq \frac{\beta_0 \mu_1}{(1-\alpha_0)\delta}. $$
We now proceed by induction on $b$.
\paragraph*{Base case ($b=0$):} The only subcube $R$ that may appear in
the inequalities \eqref{eq:one1}--\eqref{eq:one2} is the one with
empty support (that is, $R$ contains all inputs). Since
$\mu_1 \geq \frac{1}{4}$, we conclude from \eqref{eq:one1} and
\eqref{eq:one2} that $1-\alpha_1 \leq w_R \leq \beta$ or $\alpha_1
+ \beta_1 \geq 1$; so the claim holds trivially.

\paragraph*{Induction step ($b \geq 1$):} Using \eqref{eq:zero1} and
\eqref{eq:zero2} and \cref{cl:biasedcube}, we conclude that
there is a rectangle $A_0$ such that
$\mu_1(A_0) \leq \delta \mu_0(A_0)$.  We first query the bits in
the support of $A_0$. For each result $\sigma \in \{0,1\}^a$, we are
left with a subcube of inputs to investigate.

By \cref{cl:eliminate}, we have, as $B$
ranges over subcubes whose supports are disjoint from $A_0$'s, that
$$ \sum_{B} \mu_1(B) w_B  \leq \beta_1 + \delta. $$
It follows from \eqref{eq:one1} that (now summing over all $R$ whose
supports intersect $A_0$'s)
\begin{align}
\sum_R \mu_1(R)w_R &\geq (1-\alpha_1) \mu_1 - \beta_1 - \delta \nonumber\\
& \geq (1- (\alpha_1 + 4 \beta_1 + 4 \delta)) \mu_1. 
\label{eq:afterelimination}
\end{align}
For each outcome $\sigma$ for the bits queried, let $\mu^{\sigma}$ be
the resulting conditional distribution on inputs, where variables in
support of $A_0$ are fixed at $\sigma$ in $\mu^{\sigma}$. We will now
construct an LP-solution satisfying \eqref{eq:zero1}--\eqref{eq:one1}
for this derived problem. The components $u_R$ will be retained
without any change. For $w_R$, set $w_R=0$ for $R$ whose support is
disjoint from $A_0$'s, and define $\alpha_1^{\sigma}$ by
$$
 \sum_R \mu^\sigma_1(R) w_R := (1-\alpha_1^{\sigma}) \mu^\sigma_1.
$$
Then, by \eqref{eq:afterelimination}, $\mu^{\sigma}$ is an
$(\alpha_0, \beta_0, \alpha^\sigma_1,
\beta_1, a, b-1)$-product distribution for $f$ (recall our convention
that we do not include the index of a fixed bit in the
support of our subcubes). Furthermore, by \eqref{eq:afterelimination}
\begin{equation}
 \E[\alpha^{\sigma}_1] \leq \alpha_1 + 4 \beta_1 + 4 \delta.  \label{eq:onesigma}
\end{equation}
By induction, $\mu^{\sigma}$ admits a decision tree of depth at
most $a(b-1)$ that errs with probability at most
\[ \eps^{\sigma} \leq  
\frac{1}{4} +  \alpha^\sigma_1 + \beta_1 +  4(b-1)(\beta_1  + \delta) +
\frac{\beta_0}{(1-\alpha_0)\delta}
,\] when inputs are drawn according to $\mu^\sigma$. 
It follows from
\eqref{eq:onesigma}, the overall error is 
\begin{align*}
\E_\sigma[\eps^\sigma]
     & \leq \E_{\sigma}\left[ \frac{1}{4} +  \alpha^\sigma_1 + \beta_1 +  4(b-1)(\beta_1  + \delta) +
\frac{\beta_0}{(1-\alpha_0)\delta}\right]    \\
& \leq 
\frac{1}{4} +  \alpha_1 + \beta_1 +  4b(\beta_1  + \delta) +
\frac{\beta_0}{(1-\alpha_0)\delta}. \qedhere
\end{align*}
\end{proof}

\section*{Acknowledgements}

We thank Mika G{\"{o}}{\"{o}}s for pointing out an error in an
earlier version of the paper and for referring us to the literature.

{\small

}

\appendix

\section{Query Complexity LP bounds}{\label{sec:qcapp}

\begin{definition}[Query partition bound~\cite{JainK2010}]

The query partition bound of $f:\{0,1\}^n \rightarrow \{0,1\}$,
denoted $\qprt_\ve(f)$, is the optimal value of the following
linear program. In the following $A$ is a subcube of $\{0,1\}^n$.


\vspace{0.2in}

{\footnotesize
\begin{minipage}{2.2in}
    \centerline{\underline{Primal}}\vspace{-0.1in}
    \begin{align}
  &   &  {\min}\quad  \sum_{z} \sum_{A}  w_{z,A} \cdot 2^{|A|} \\
& {\rm s.t.} & \nonumber \\ 
       \quad &  \forall x\in\{0,1\}^n  :&  \sum_{A: x \in A} w_{f(x),A} \geq 1 - \epsilon \label{eq:constr1}\\
      & \forall x\in \{0,1\}^n :&  \sum_{A: x \in A} \quad \sum_{z}  w_{z,A} = 1  \label{eq:constr2}\\
      & \forall (z,A)  :&  w_{z,A} \geq 0 \label{eq:constr3}
         \end{align}
\end{minipage}\hspace{0.5in}
\begin{minipage}{2.2in}\vspace{-1in}
    \centerline{\underline{Dual}}\vspace{-0.1in}
        \begin{align*}
    & & \max \quad    (1-\epsilon)\sum_{x} \mu_{x} + \sum_{x}\phi_{x} \\
& {\rm s.t.} & \\ 
      &  \forall (z,A) :& \sum_{x\in A \cap f^{-1}(z)} \mu_{x}   + \sum_{x\in  A} \phi_{x}   \leq 2^{|A|}\\
 	& \forall x :& \mu_{x} \geq 0, \phi_{x} \in \reals
    \end{align*}
\end{minipage}
}
\end{definition}

\vspace{0.2in}



The error parameter for partition bound can be reduced as in the
following claim.

\begin{claim}\label{claim:qprt-err-red}
Let $\ve > \delta >0$. Then $\log \qprt_\delta(f) \leq O\left(\frac{1}{(0.5 - \ve)^2}\log\frac{1}{\delta} \right) \log \qprt_\ve(f) .$
\end{claim}

\begin{proof}
Let $\{ w_{z,A}\}$ be an optimal solution for the primal of
$\qprt_\ve(f)$. Using this solution, we will obtain a solution
$\{v_{z,A}\}$ showing that $\qprt_\delta(f)$ is small.
Let $t$ be an odd positive integer. 
For $z \in \{0,1\}$, let $G_z = \{(z_1, \ldots, z_t): \sum_i z_i \geq t/2\}$ and for a subcube $A$, let $G_A = \{
(A_1, \ldots, A_t): \bigcap_i A_i = A \}$.
Now consider the following assignment to the variables:
$$v_{z,A} :=\sum_{(z_1, \ldots, z_t) \in G_z, (A_1, \ldots, A_t) \in
 G_A} \prod_{i=1}^t w_{z_i,A_i} . $$ We will choose $t$ such that this
 assignment constitutes a valid solution to the above linear program 
with $\delta$ instead of $\epsilon$. 

Fix $x$.  Then,
\begin{align*}
\sum_{A \ni x} v_{f(x),A} & =   \sum_{A \ni x}  \sum_{(z_1, \ldots, z_t) \in G_{f(x)}, (A_1, \ldots, A_t) \in G_A}  \prod_{i=1}^t w_{z_i,A_i}  \\
& = \sum_{(z_1, \ldots, z_t) \in G_{f(x)} , (A_1 \ni x , \ldots, A_t \ni x) } 
 \prod_{i=1}^t w_{z_i,A_i} \\
& \geq 1 - \exp(-2(0.5-\ve)^2 t).   \quad \mbox{(using Chernoff bounds)}
\end{align*}
By choosing an appropriate $t =
O\left(\frac{1}{(0.5-\ve)^2} \ln \frac{1}{\delta}\right)$, we ensure
that constraint~\eqref{eq:constr1} is satisfied, with $\delta$ instead
of $\ve$. Furthermore,
\begin{align*}
\sum_z \sum_{A \ni x} v_{z,A} & =   \sum_z \sum_{A \ni x}  \sum_{(z_1, \ldots, z_t) \in G_z, (A_1, \ldots, A_t) \in G_A} \left( \Pi_{i=1}^t w_{z_i,A_i} \right) \\
& = \sum_{(z_1, \ldots, z_t) , (A_1 \ni x , \ldots, A_t \ni x) } \left( \Pi_{i=1}^t w_{z_i,A_i} \right) \\
& = \prod_{i=1}^t   \sum_{z_i} \sum_{A_i \ni x}      w_{z_i, A_i}  \\
& = 1 .
\end{align*}
Thus, the constraint~\eqref{eq:constr2}) is satisfied. Constraint~\eqref{eq:constr3} is clearly satisfied.

It remains the estimate the new objective value.
\begin{align*}
\sum_z \sum_{A} v_{z,A} 2^{|A|} 
&\leq \sum_z \sum_{A}  \sum_{(z_1, \ldots, z_t) \in G_z, (A_1, \ldots, A_t) \in G_A} \left( \prod_{i=1}^t w_{z_i,A_i} 2^{|A_i|}\right) \\
& = \sum_{(z_1, \ldots, z_t) , (A_1, \ldots, A_t) } \left( \prod_{i=1}^t w_{z_i,A_i} 2^{|A_i|}\right) \\
& = \prod_{i=1}^t \left(  \sum_{z_i} \sum_{A_i}      w_{z_i, A_i} 2^{|A_i|}         \right) .
\end{align*}
Thus, the right hand side is at most $\qprt_\ve(f)^t$.
Our claim follows by taking logs. 
\end{proof}

\section{Communication Complexity LP bounds}\label{sec:ccapp}

\paragraph*{Information Complexity} The notion of information complexity 
was formalized by Chakrabarti, Shi, Wirth and
Yao~\cite{ChakrabartiSWY2001} in the direct sum context.  Similar
information theoretic arguments have been used in earlier works as well without explicitly defining information complexity 
(see Nisan and Wigderson~\cite{NisanW1993} and Ponzio,
Radhakrishnan and Venkatesh~\cite{PonzioRV2001} for applications to
the pointer-chasing
problem). Chakrabarti~\etal~\cite{ChakrabartiSWY2001} defined and
used, what in today's language is called,``external information
cost''. Bar-Yossef, Jayram, Kumar and Sivakumar~\cite{BarYossefJKS2004} used
the notion, what in today's language is called the``internal information
cost'', in their proof of the disjointness lower bound.
The formal definition and the terminology ``internal information
cost'' and ``external information cost'' was introduced by
Barak~\etal~\cite{BarakBCR2013}.

\begin{definition}[Information Complexity] Let $\Pi$ be any randomized
  protocol between two parties Alice and Bob with inputs $\X \times \Y$ that
  are distributed according to a distribution $\mu$. The (internal)
  information cost of the protocol $\Pi$, denoted $\IC^\mu(\Pi)$ is
  defined as $$\IC^\mu(\Pi) = I[X: T_\Pi | Y ] + I[Y:T_\Pi|X],$$
where $(X,Y)$ is the random variable denoting the pair of inputs and
$T_\Pi$ is the random variable denoting the transcript of the protocol
(including public randomness). 

For a function $f: \X
  \times \Y \rightarrow \Z$, $\ve \in (0,1)$ and a distribution
  $\mu$ on $\X \times \Y$, the $\ve$-information complexity of $f$
  under $\mu$, denoted $\IC^\mu_\ve(f)$, is defined to be $\min_\Pi
  \IC^\mu(\Pi)$, where the minimum is taken over all protocols $\Pi$ that compute
  $f$ with error at most $\ve$ under the distribution $\mu$. The
  $\ve$-information complexity of $f$, denoted $\IC_\ve(f)$, is
  defined to be $\max_\mu \IC^\mu_\ve(f)$, where the maximum is over
  all distributions $\mu$.
\end{definition}
An alternate notion of information complexity was considered by
Braverman~\cite{Braverman2015} by changing the order of quantifiers.
\begin{definition}[Prior-free information complexity]
For a function $f: \X
  \times \Y \rightarrow \Z$ and $\ve \in (0,1)$ , the prior free
  information complexity, denoted $\IC^p_\ve(f)$, is defined to be
  $\min_\Pi \max_\mu \IC^\mu(\Pi)$ where the minimum is taken over all
  randomized protocols $\Pi$ that compute the function $f$ correctly
  on all inputs with error at most $\ve$ and the maximum is over all
  distributions $\mu$.
\end{definition}
Braverman~\cite{Braverman2015} showed that these two notions of
information complexity are within constant factors of each other. More
precisely,
$$\IC_\ve(f) \leq \IC^p_\ve(f) \leq 2 \cdot \IC_{\ve/2}(f).$$

This alternate characterization of information complexity immediately
yields the following error reduction claim.
\begin{claim}\label{claim:IC-err-red}
Let $\delta >0$. Then $\IC_\delta(f) \leq O(\log\frac{1}{\delta})  \cdot  \IC_{1/8}(f) .$
\end{claim}

\paragraph*{The partition bound:} The partition bound was introduced by Jain and Klauck~\cite{JainK2010} as a linear programming bound to lower bound the public-coin randomized communication complexity.

\begin{definition}[Partition bound~\cite{JainK2010}] For a function $f: \X \times \Y \rightarrow \Z$ and $\ve \in (0,1)$, 
the $\ve$-partition bound of $f$, denoted $\qprt_\ve(f)$, is defined to be the  optimal value of the following linear program. Below $R$ represents a rectangle in $\X\times \Y$ and  $(x,y,z) \in \X \times \Y \times \Z$.\\

{\footnotesize
\begin{minipage}{2.2in}
    \centerline{ \underline{Primal}}
 \begin{align*}
    {\min}\quad \sum_{z} \sum_{R}  w_{z,R} \\
    \sum_{R: (x,y) \in R} w_{f(x,y),R} &\geq 1 - \epsilon, && \forall (x,y)\\
    \sum_{R: (x,y) \in R} \quad \sum_{z}  w_{z,R} &= 1, && \forall (x,y)  \\
      w_{z,R} &\geq 0, &&  \forall (z,R) \enspace .
    \end{align*}
\end{minipage}\hspace{0.5in}
\begin{minipage}{2.2in}\vspace{-0.1in}
    \centerline{\underline{Dual}}\vspace{-0.0in}
 \begin{align*}
 {\max}\quad    (1-\epsilon)\sum_{(x,y)} \mu_{x,y} + \sum_{(x,y)}\phi_{x,y} \\
      \sum_{(x,y)\in R \cap f^{-1}(z)} \mu_{x,y}   + \sum_{(x,y)\in  R} \phi_{x,y}  & \leq 1, && \forall (z,R)\\
      \mu_{x,y} &\geq 0, && \forall (x,y)\\
      \phi_{x,y} &\in \reals, && \forall (x,y) \enspace . 
 \end{align*}
\end{minipage}
}
\end{definition}
Jain and Klauck~\cite{JainK2010} show that $\CC^{\rm pub}_\ve(f) \geq \log(\prt_\ve(f))$.

The error in the partition bound in the communication complexity
setting can be reduced in a fashion similar to query setting (ie., \cref{claim:qprt-err-red}). 
\begin{claim}\label{claim:prt-err-red}
Let $\delta >0$. Then $\log \prt_\delta(f) \leq
O(\log\frac{1}{\delta})  \cdot  \log \prt_{1/8}(f) .$
\end{claim}

\paragraph*{Relaxed partition bound:} Kerenidis~\etal~\cite{KerenidisLLRX2015} defined a relaxation of the partition bound by relaxing the equality constraint in the primal.

\begin{definition}[Relaxed partition bound~\cite{KerenidisLLRX2015}] For a function $f: \X \times \Y \rightarrow \Z$ and $\ve \in (0,1)$, 
the $\ve$-relaxed partition bound of $f$, denoted $\rprt_\ve(f)$, is given by the  optimal value of the following linear program. \\

{\footnotesize
\begin{minipage}{2.2in}
    \centerline{ \underline{Primal}}
 \begin{align*}
    {\min}\quad \sum_{z} \sum_{R}  w_{z,R} \\
    \sum_{R: (x,y) \in R} w_{f(x,y),R} &\geq 1 - \epsilon, && \forall (x,y)\\
    \sum_{R: (x,y) \in R} \quad \sum_{z}  w_{z,R} &\leq 1, && \forall (x,y)  \\
      w_{z,R} &\geq 0, &&  \forall (z,R) \enspace .
    \end{align*}
\end{minipage}\hspace{0.5in}
\begin{minipage}{2.2in}\vspace{-0.2in}
    \centerline{\underline{Dual}}\vspace{-0.0in}
 \begin{align*}
 {\max}\quad    (1-\epsilon)\sum_{(x,y)} \mu_{x,y} - \sum_{(x,y)}\phi_{x,y} \\
      \sum_{(x,y)\in R \cap f^{-1}(z)} \mu_{x,y}   - \sum_{(x,y)\in  R} \phi_{x,y}  & \leq 1, && \forall (z,R)\\
      \mu_{x,y} &\geq 0, && \forall (x,y)\\
      \phi_{x,y} &\geq 0, && \forall (x,y)\enspace . 
 \end{align*}
\end{minipage}
}
\end{definition}
Clearly, $\CC^{\rm pub}_\ve(f) \geq \log(\prt_\ve(f))\geq
\log(\rprt_\ve(f)) \geq \log
\srec_{\epsilon,\epsilon}(f)$. Kerenidis~\etal~\cite{KerenidisLLRX2015}
showed that relaxed partition bound is upper bounded by information
complexity. Combining these facts with the error reduction claims for
partition bound and information complexity
(\cref{claim:IC-err-red,claim:prt-err-red}) yields \cref{prop:sr_rprt}.





\end{document}